\documentclass[runningheads]{llncs}
\usepackage{amsmath}
\usepackage{braket}
\usepackage{quantikz}
\usepackage{booktabs}
\usepackage{url}
\usepackage[left=3.5cm,right=2.5cm,top=2.5cm,bottom=2.5cm]{geometry}
\usepackage{orcidlink}

\newcommand{\rx}[2]{{\left(\begin{array}{cc}#1 & #2 \\ #2 & #1\end{array}\right)}}
\newcommand{\rz}[2]{{\left(\begin{array}{cc}#1 & 0 \\ 0 & #2\end{array}\right)}}
\newcommand{\matrixtwo}[4]{\left(\begin{array}{cc}#1 & #2 \\ #3 & #4\end{array}\right)}

\begin{document}

\title{The Quantum Approximate Optimization Algorithm\\Can Require Exponential Time to Optimize Linear Functions\thanks{This preprint has been accepted in the Quantum Optimization Workshop at the Genetic and Evolutionary Computation Conference (GECCO 2025). The accepted version can be found at \doi{10.1145/3712255.3734319}. This research is partially funded by the University of Malaga under grant PAR 4/2023.}}

\author{Francisco Chicano\inst{1}\orcidlink{0000-0003-1259-2990} \and
Zakaria Abdelmoiz Dahi\inst{2}\orcidlink{0000-0001-8022-4407} \and
Gabriel Luque\inst{1}\orcidlink{0000-0001-7909-1416}
} 
\authorrunning{F. Chicano, Z. Dahi, and G. Luque}
\institute{ITIS Software, University of Malaga, Spain 
\and
Univ. Lille, Inria, CNRS, Centrale Lille, UMR 9189 CRIStAL, F-59000 Lille, France\\
\email{chicano@uma.es, abdelmoiz-zakaria.dahi@inria.fr, gluque@uma.es}
}

\maketitle

\begin{abstract}
QAOA is a hybrid quantum-classical algorithm to solve optimization problems in gate-based quantum computers. It is based on a variational quantum circuit that can be interpreted as a discretization of the annealing process that quantum annealers follow to find a minimum energy state of a given Hamiltonian. This ensures that QAOA must find an optimal solution for any given optimization problem when the number of layers, $p$, used in the variational quantum circuit tends to infinity. In practice, the number of layers is usually bounded by a small number. This is a must in current quantum computers of the NISQ era, due to the depth limit of the circuits they can run to avoid problems with decoherence and noise. In this paper, we show mathematical evidence that QAOA requires exponential time to solve linear functions when the number of layers is less than the number of different coefficients of the linear function $n$. We conjecture that QAOA needs exponential time to find the global optimum of linear functions for any constant value of $p$, and that the runtime is linear only if $p \geq n$. We conclude that we need new quantum algorithms to reach quantum supremacy in quantum optimization.

\end{abstract}



\section{Introduction}

Quantum computing is a computational paradigm based on quantum mechanical properties such as superposition and entanglement, which is thought to provide a computational speedup over classical computing~\cite{google2019}. Such a computational acceleration could have a wide range of applications, among which we find problem optimization. Several techniques have been devised, but due to the Noisy Intermediate Scale (NISQ) nature of nowadays' quantum hardware, the feasibility and reliability of the existing quantum optimizers are still far from ideal. 

One of the most popular algorithms for optimization designed for gate-based machines is the Quantum Approximate Optimization Algorithm (QAOA) \cite{farhi14-qaoa}. It is a hybrid quantum-classical algorithm, where the quantum part is a parameterized quantum circuit composed of different layers with the same structure, the so-called \emph{ansatz}, while the classical part is an optimization algorithm to tune the ansatz parameters. Mathematically speaking, QAOA is a slight variant of the discretized version of the \emph{quantum annealing} process that quantum annealers (like the ones of D-Wave) do to optimize a Hamiltonian~\cite{PhysRevX.10.021067}. This ensures that as the number of layers, $p$, in the ansatz tends to infinity, the QAOA approximates the continuous path of the annealing process of quantum annealers, reaching a global optimal solution to the optimization problem. 

There are many works on QAOA and different variants have been proposed since 2014. A search in Scopus produces 773 papers with the terms ``QAOA'' or its full name in the title, abstract, or keywords. A recent survey~\cite{ref8} identifies 17 variants of QAOA (see Table 1 in the cited paper). The survey highlights that it is not clear in which problems QAOA has an advantage over classical algorithms, but there are some particular problems for which it is already known that QAOA cannot outperform classical algorithms. One example is the family of MaxCut problems for bipartite $D$-regular graphs.
On the positive list of achievements of QAOA, we can highlight the quantum supremacy result of Farhi and Harrow~\cite{ref9}, where the authors prove that the probability distribution obtained by QAOA for a particular ansatz with one layer ($p=1$) cannot be reproduced efficiently by a classical algorithm.

While we see that QAOA can provably outperform classical algorithms for some particular tasks, these tasks are not always aligned with the ones we would like to solve using quantum computers. Taking as an example the mentioned supremacy of QAOA, the probability distribution it can efficiently obtain (compared to a classical algorithm) is probably not very relevant from a practical point of view. In this paper, in particular, we are interested in solving optimization problems, that is, finding solutions that maximize or minimize an objective function. This is the main task for which QAOA was designed. Thus, we are interested in evaluating the performance of QAOA when applied to optimization problems. 

After revising some background in Section~\ref{sec:background}, we provide some initial results on this analysis in Section~\ref{sec:proposal} using the simplest objective functions we can imagine: linear functions. If QAOA faces problems in optimizing linear functions, many of the more complex families of problems could also be challenging for QAOA. For example, the family of Quadratic Unconstrained Binary Optimization (QUBOs) contains linear functions as a particular case. We find that QAOA can require an exponential number of measurements (and time) to find the optimal solution to linear functions. We prove this result for low-depth QAOA where $p=1$ or $p=2$.
We discuss the consequences of this result in Section~\ref{sec:discussion}, where we also provide ideas for designing potential quantum algorithms that could outperform classical algorithms in optimization. Section~\ref{sec:conclusions} concludes the paper pointing to future research directions.

% Complete paragraph and mention Background and Conclusions.

%Considering the above facts, the contribution of this work stands in proving that the QAOA requires an exponential time to optimize linear functions. This has been done through a rigorous theoretical and practical runtime analysis, establishing \texttt{(I)} that such implication is true when the number of layers (i.e. discrete time steps) of the QAOA is less than the number of coefficients of the linear function, and \texttt{(II)} the QAOA's runtime is linear only in the inverse scenario. The findings of this work aims at highlighting the limitations of today's NISQ algorithms, and also provide a better understanding of the type/features that quantum optimisation algorithms should have to attain quantum supremacy. 

%The remainder of the paper is structured as follows. In Section \ref{sec:background}, the fundamentals of quantum computing/optimisers and mathematics are introduced. Afterwards, Section \ref{sec:proposal} presents the core proposal of this work by analyzing, theoretically and empirically, the runtime of QAOA for linear functions. Finally, Section \ref{sec:discussion} discusses the results of the analysis, while Section \ref{sec:conclusions} concludes the paper.

\section{Background}
\label{sec:background}
This section presents some mathematical preliminaries, as well as some fundamentals on quantum computing and QAOA. 

\subsection{Mathematical preliminaries}
\label{subsec:maths}

We will use the notation $[n]=\{1,\ldots,n\}$ to refer to a set of consecutive integers. If we want to emphasize that the consecutive numbers form a tuple we use $(n)=(1,2,\ldots,n)$. Dirac notation will be used along the paper, where \emph{kets}, denoted with $\ket{\psi}$, are \emph{column vectors} and \emph{bras}, denoted with $\bra{\varphi}$, are \emph{row vectors}. \emph{Matrices} (or \emph{linear operators}) will be represented with latin capital letters. In general, vectors and matrices are complex-valued. We denote the \emph{imaginary unit} with $i = \sqrt{-1}$ along the paper and will avoid the use of the lowercase latin letter $i$ for indices to prevent any confusion. If $a$ is a complex number, we can write it as $a= |a| e^{i \alpha}$, where $|a|$ is the \emph{modulus} of $a$ and $\alpha$ its phase. The \emph{complex conjugate} of $a$ is denoted with $a^* = |a|e^{-i \alpha}$. The \emph{conjugate transpose} of a matrix $A$ will be denoted with $A^{\dagger}$ and is computed from $A$ by transposing it and using the complex conjugate of its coefficients, that is, $A^\dagger = (A^T)^*$. In Dirac's notation, the conjugate transpose of a bra is a ket denoted with the same symbol and vice versa. That is, if $\ket{\psi}=(a, b, c, \ldots)^T$, then $\bra{\psi}=(a^*,b^*,c^*, \ldots)$. The \emph{dot product} of bra $\bra{\varphi}$ and ket $\ket{\psi}$ is denoted with $\braket{\varphi|\psi}$. The \emph{norm} of a ket $\ket{\psi}$ is $\lVert\ket{\psi} \rVert=\sqrt{\braket{\psi|\psi}}$. We will work mainly with \emph{normalized vectors}, which have norm 1. In this paper, we will only consider complex-valued finite vector spaces with a dot product. All these spaces are \emph{complete} (the elements of any Cauchy sequence converge to an element of the space) and are called \emph{Hilbert spaces}. 

We say that a matrix $U$ is \emph{unitary} if it satisfies $U U^{\dagger} = U^{\dagger} U = I$, where $I$ is the identity matrix. This means that the inverse of $U$ is $U^{-1}=U^{\dagger}$. We say that matrix $H$ is \emph{Hermitian} if $H=H^{\dagger}$. Two matrices, $A$ and $B$, \emph{commute} if $A B = B A$. The \emph{commutator} of $A$ and $B$ is $[A,B]=A B-B A$.
If $A$ is a square matrix, the \emph{matrix exponential} $e^A$ is defined based on the power series
\begin{equation}
    e^A = \sum_{j=0}^{\infty} \frac{1}{j!} A^j ,
\end{equation}
where $A^0=I$, the identity matrix.
In general, it is not true that $e^{A+B}=e^A e^B$, but if $A$ and $B$ commute the equality holds. If $H$ is a Hermitian matrix, then $e^{-i H}$ is unitary.

Given two vector spaces $V$ and $W$, the \emph{tensor product}, denoted with $V \otimes W$ is the vector space composed of terms with the form $\ket{\varphi} \otimes \ket{\psi}$ where $\ket{\varphi} \in V$, $\ket{\psi} \in W$, and the tensor product $\otimes$ applied to vectors is a multi-linear operation. 
If $A$ is a linear operator defined in vector space $V$ and $B$ is a linear operator defined in vector space $W$, then $A \otimes B$ is a linear operator acting on the tensor product space $V \otimes W$ as $(A \otimes B)(\ket{\varphi} \otimes \ket{\psi})=(A \ket{\varphi}) \otimes (B \ket{\psi})$. As a consequence, if $A$ and $C$ are linear operators in $V$ and $B$ and $D$ are linear operators in $W$ we have $(A\otimes B)(C\otimes D)=(AC) \otimes (BD)$.

\subsection{Gate-based quantum computing}
\label{subsec:qc}

A \emph{qubit} is a quantum mechanical system with two states that can be manipulated. It is the minimum building block for quantum computations.
Mathematically, the two orthogonal basis states we can observe after measuring a qubit are represented by vectors $\ket{0}$ and $\ket{1}$, and they belong to a two-dimensional complex vector space which forms also a \emph{Hilbert space} with the dot product defined in Section~\ref{subsec:maths}. In general, the state of a qubit is represented by a vector $\ket{\psi} = \alpha \ket{0}+\beta \ket{1}$, where $\alpha$ and $\beta$ are the (complex-valued) probability amplitudes of the states $\ket{0}$ and $\ket{1}$. If we measure a qubit in this state, the values 
$|\alpha|^{2}$ and $|\beta|^{2}$ give the probabilities of observing state 0 or 1, respectively, and they fulfill $|\alpha|^{2} + |\beta|^{2}=1$. Once the qubit is measured, it will \emph{collpase} to the measured basis state: $\ket{0}$ or $\ket{1}$. The quantum state of $n$ qubits is represented with the tensor product of $n$ Hilbert spaces representing one qubit.

Several quantum computing paradigms exist, although in this work we focus on the discrete gate-based quantum computers~\cite{ref1}. This paradigm describes quantum computations as quantum circuits, which can be seen as a series of quantum gates acting on a set of qubits. A quantum gate is a unitary transformation $U$ acting on one or more qubits. 
The single-qubit quantum gates relevant in this paper are the Hadamard gate \texttt{H}, the $X$ rotation gate \texttt{RX} and the $Z$ rotation gate \texttt{RZ}. These three gates together with the two-qubits \texttt{CNOT} gate form a \emph{universal set of gates}\footnote{In fact, the Hadamard gate can be built with the $X$ and $Z$ rotations.}. The Pauli \texttt{X} and \texttt{Z} gates are $X$ and $Z$ rotations with angle $\pi$ (times an irrelevant phase constant). The Hadamard, Pauli \texttt{X}, and Pauli \texttt{Z} gates are Hermitian and their own inverses: $\texttt{H}\texttt{H}=\texttt{X}\texttt{X}=\texttt{Z}\texttt{Z}=I$. We present below the mathematical definition of these gates in matrix form:

\begin{eqnarray}
    && \texttt{H} = \frac{1}{\sqrt{2}} \matrixtwo{1}{1}{1}{-1}, \\
    && \texttt{RX}(\theta) = e^{-i \frac{\theta}{2} \texttt{X}} =  \rx{\cos \left(\theta/2\right)}{-i \sin \left(\theta/2\right)}, \\
    && \texttt{RZ}(\theta) = e^{-i \frac{\theta}{2} \texttt{Z}} = \rz{e^{-i \theta/2}}{e^{i \theta/2}}, \\
    && \texttt{CNOT} = \left(\begin{array}{cccc}
    1 & 0 & 0 & 0 \\
    0 & 1 & 0 & 0 \\
    0 & 0 & 0 & 1 \\
    0 & 0 & 1 & 0
    \end{array}\right) .
\end{eqnarray}

% \begin{minipage}{\linewidth}\centering
%    \hspace{-3em}
%     \begin{minipage}{0.45\linewidth}
%     \begin{equation}
%     \texttt{H} = 
%     \begin{bmatrix}
%     1/\sqrt{2} & 1/\sqrt{2} \\
%     1/\sqrt{2} & -1/\sqrt{2}
%     \end{bmatrix}
%      \label{eq:qg1}
%     \end{equation}
%       \end{minipage}
%     % \hspace{0.05\linewidth}
%       \begin{minipage}{0.45\linewidth}
%     \begin{equation}
%     \texttt{CNOT} = 
%     \begin{bmatrix}
%     1 & 0 & 0 & 0 \\
%     0 & 1 & 0 & 0 \\
%     0 & 0 & 0 & 1 \\
%     0 & 0 & 1 & 0
%     \end{bmatrix}
%      \label{eq:qg2}
%     \end{equation}
%       \end{minipage}

% \end{minipage}
% \vspace{1em}

% \begin{minipage}{\linewidth}\centering
%    \hspace{-3em}
%     \begin{minipage}{0.45\linewidth}
%     \begin{equation}
%     \texttt{RZ}(\theta) = 
%     \begin{bmatrix}
%     e^{-i\frac{\theta}{2}} & 0 \\
%     0 & e^{i\frac{\theta}{2}}
%     \end{bmatrix}
%     \label{eq:rz}
%     \end{equation}
%       \end{minipage}
%       \hspace{0.05\linewidth}
%     \begin{minipage}{0.57\linewidth}
%     \begin{equation}
%     \texttt{RX}(\theta) = 
%     \begin{bmatrix}
%     cos(\frac{\theta}{2}) & -i sin(\frac{\theta}{2}) \\
%     -i sin(\frac{\theta}{2}) & cos(\frac{\theta}{2})
%     \end{bmatrix}
%     \label{eq:rx}
%     \end{equation}
%   \end{minipage}
% \end{minipage}
% \vspace{1em}

Using the exponential matrix definition of \texttt{RX} and \texttt{RZ} it is easy to see that $\texttt{RX}(\theta_1)\texttt{RX}(\theta_2) = \texttt{RX}(\theta_1 + \theta_2)$ and $\texttt{RZ}(\theta_1)\texttt{RZ}(\theta_2) = \texttt{RZ}(\theta_1 + \theta_2)$. This implies that \texttt{RX} gates commute among them and \texttt{RZ} gates also commute. It also implies that the inverse of $\texttt{RX}(\theta)$ is $\texttt{RX}(-\theta)$ and the same applies to the \texttt{RZ} gate. Vectors $\ket{0}$ and $\ket{1}$ are \emph{eigenvectors} of the \texttt{Z} gate with \emph{eigenvalues} $+1$ and $-1$, respectively. That is, $\texttt{Z} \ket{0} = \ket{0}$ and $\texttt{Z} \ket{1} = - \ket{1}$. The eigenvectors of the \texttt{X} gate also form a basis of the 1-qubit Hilbert space and they are denoted with $\ket{+}$ and $\ket{-}$, with eigenvalues $+1$ and $-1$, respectively. The Hadamard gate transforms one of the basis into the other: $\texttt{H} \ket{0} = \ket{+}$ and $\texttt{H} \ket{1} = \ket{-}$.

\subsection{Quantum optimizers and QAOA}
\label{subsec:qaoa}

Quantum technology is still in very early stages in terms of stability and scalability. Considering the Noisy-Intermediate Scale (NISQ) state of today's quantum hardware, it is very tricky to execute complex and fault-tolerant calculations. Bearing in mind the noise and qubits' limitations, \emph{Variational Quantum Algorithms} (VQAs) is a class of quantum algorithms that is showing feasible applicability, although running on NISQ machines. The VQAs rely on a layered application of a quantum ansatz (see Figure \ref{fig:vqe}). Depending on the design and purpose of the algorithm, the mathematical description of each ansatz might differ. For instance, when analyzing the \emph{variational quantum classifier}, usually part of the ansatz is responsible for data encoding, while the second part deals with the classification task \cite{ref2}. Another example quite related to this work is the \emph{Variational Quantum Eigensolver} (VQE) originally proposed in \cite{ref3}. An VQE is a hybrid quantum-classical optimizer that relies on both a quantum and a classical algorithm (Figure~\ref{fig:vqe}). The quantum part of the VQE is a parameterized quantum ansatz $U$($\Theta$) whose parameters $\Theta = (\theta_1, \dots, \theta_n)$ are evolved using a classical optimizer that optimizes a Hamiltonian $H$.

\begin{figure}[ht!]
    \centering
    \includegraphics[width=0.72\linewidth]{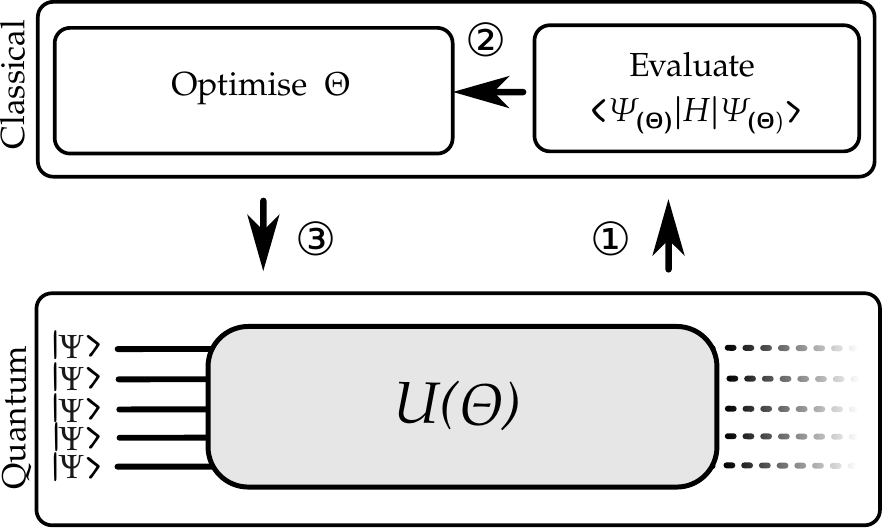}
%    \Description[Scheme of variational quantum algorithm.]{Scheme of variational quantum algorithm.}
    \caption{Variational quantum algorithm.}
    \label{fig:vqe}
\end{figure}

One of the VQAs for combinatorial optimization is the \emph{Quantum Approximate Optimization Algorithm} (QAOA) \cite{farhi14-qaoa}, a hybrid quantum-classical optimizer where the quantum part is composed of $p$ layers, each one containing the so-called \emph{problem unitary} and \emph{mixer unitary} (see Figure~\ref{fig:qao}). These unitary operations are defined as follows:
\begin{eqnarray}
    && U(H_{P}, \gamma_{j}) = e^{-i \gamma_j H_P} , \\
    && U(H_{M}, \beta_{j}) = e^{-i \beta_j H_M} ,
\end{eqnarray}
where $H_P$ is the \emph{problem Hamiltonian}, which is the objective function we want to optimize; and $H_M$ is the mixer Hamiltonian, which is a sum of \texttt{X} gates, one for each qubit:
\begin{equation}
    H_M = \sum_{j=1}^n \texttt{X}_j \; .
\end{equation}

\begin{figure}[ht!]
    \centering
    \includegraphics[width=0.8\linewidth]{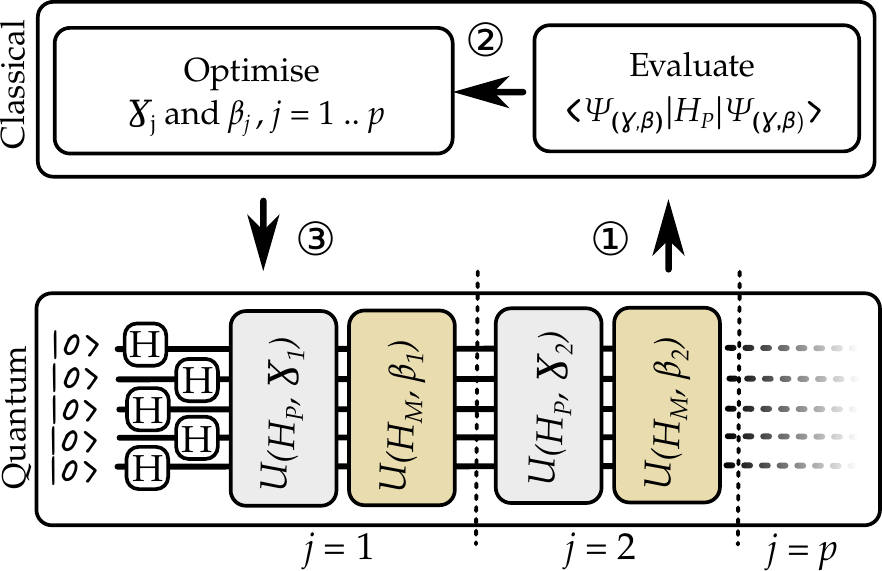}
%    \Description[Scheme of QAOA.]{Scheme of QAOA.}
    \caption{The QAOA's execution workflow}
    \label{fig:qao}
\end{figure}

The parameters of the ansatz are the two vectors $\gamma=(\gamma_{1}, \dots ,\gamma_{p})$ and $\beta=(\beta_{1}, \dots ,\beta_{p})$. The state obtained after applying the QAOA ansatz is given by
\begin{equation}
\label{eq:psi}
    \ket{\psi(\gamma, \beta)} = \left(\prod_{j=p}^1 U(H_{M}, \beta_{j}) U(H_{P}, \gamma_{j}) \right) \texttt{H}^{\otimes n} \ket{0\ldots 0} .
\end{equation}

QAOA relies on a classical optimizer to find the vectors $\gamma$ and $\beta$ that optimize the value of 
\begin{equation}
\label{eq:ev}
    E{(\gamma,\beta)} = \bra{\psi{(\gamma,\beta)}} H_P \ket{\psi{(\gamma,\beta)}}.
\end{equation}

% Thanks to the commutativity of the Pauli-Z and Pauli-X gates, the unitary transformations $U(H_{P}, \gamma_{k})$ and $U(H_{M}, \beta_{k})$ can be described as the product of the exponential of the Pauli-Z and Pauli-X terms (see Equations \eqref{eq:up} and \eqref{eq:um}). Depending on the order of the Pauli-Z polynomials describing the problem Hamiltonian, the exponential of Pauli-Z gates in Equation \eqref{eq:up}, results in RZ gates for linear terms (see Equation \eqref{eq:rz}), and RZZ gates for quadratic ones, and higher order Z rotation gates for higher order terms in the Hamiltonian. The same goes for the exponential of Pauli-X gates in Equation \eqref{eq:um}, which will result in RX gates (see Equation \eqref{eq:rx}). 

% \begin{equation} \label{eq:up}
%     \resizebox{.9\hsize}{!}{$
%     U(H_{P}, \gamma_{j}) = e^{-i \gamma_j \sum\limits_{S \in [n]} a_S \bigotimes\limits_{k \in S} Z_k}
%     = \prod\limits_{S \in [n]} e^{-i \gamma_j a_S \bigotimes\limits_{k \in S} Z_k}
%     $}
% \end{equation}

% \begin{equation} \label{eq:um}
%     U(H_{M}, \beta_{j}) = e^{-i \beta_{j} \sum\limits_{k=1}^{n}X_{k}} = \prod_{k=1}^{n} e^{-i \beta_{j} X_{k}}
% \end{equation}

\section{The runtime of QAOA for linear functions}
\label{sec:proposal}

In this section we present the main result of the paper: the runtime analysis of QAOA for linear functions. 
We assume the goal is to maximize the linear Ising model
\begin{equation}
\label{eqn:linear-ising}
    H_P(s) = \sum_{\ell=1}^{n} a_\ell s_\ell,
\end{equation}
where $a_\ell$ are real nonzero coefficients and $s_\ell$ are spin variables taking values in $\{-1,1\}$. The solution maximizing Equation~\eqref{eqn:linear-ising} is
\begin{equation}
\label{eqn:solution-linear}
    s_\ell=sign(a_\ell) \; ,
\end{equation}
for all $\ell \in [n]$, where $sign(a)$ is $-1$ for $a<0$ and $1$ for $a>0$. Knowing the optimal solution to the linear Ising model simplifies the computation of the exact probability for the ansatz to reach the global optimum.

As explained in Section~\ref{subsec:qaoa}, in QAOA there is a classical algorithm optimizing the average value of the problem Hamiltonian $H_P$, that is, Equation~\eqref{eq:ev}. For each iteration of this classical algorithm, the ansatz has to be measured several times, contributing to the runtime of QAOA. 
Thus, the runtime not only depends on the ansatz design, but also on the classical optimization algorithm. This makes it difficult to do a precise analysis of QAOA.
For the sake of simplicity we will omit the effect of the classical optimization algorithm here. 
We assume that in the classical part there is an \emph{oracle} which provides the optimal values for the ansatz parameters: the vectors $\gamma$ and $\beta$. In practice, this means that the runtime we will obtain is a \emph{lower bound} of the real runtime of QAOA for any classical optimization algorithm. The fact that we compute a lower bound instead of the precise runtime will not change the main conclusions of this paper because, as we will see shortly, this lower bound is exponential  in the number of qubits $n$ when $p$ is bounded by a constant. With the optimal $\gamma$ and $\beta$ parameters provided by the oracle, we only need to sample the ansatz to get the solutions. We will compute the probability of obtaining the optimal solution by sampling the state $\ket{\psi(\gamma, \beta)}$ provided by Equation~\eqref{eq:psi}. The inverse of this probability is the average number of samples we need to reach the global optimal solution (according to the geometric probability distribution).

Let us start by developing Equation~\eqref{eq:psi} taking into account that the problem Hamiltonian, $H_P$, is the linear Ising model of Equation~\eqref{eqn:linear-ising}:
\begin{align}
\label{eqn:psi-linear}
    \ket{\psi(\gamma, \beta)} &= \left(\prod_{j=p}^1 e^{-i \beta_j H_M} \cdot e^{-i \gamma_j H_P}\right) \texttt{H}^{\otimes n} \ket{0\ldots0} \\
\nonumber    &= \left(\prod_{j=p}^1 e^{-i \beta_j \sum_{\ell=1}^n \texttt{X}_\ell} \cdot e^{-i \gamma_j \sum_{\ell=1}^{n} a_\ell \texttt{Z}_\ell}\right) \bigotimes_{\ell=1}^n \ket{+}_\ell \\
\nonumber    &= \left(\prod_{j=p}^1 \left(\bigotimes_{\ell=1}^{n} e^{-i \beta_j \texttt{X}_\ell}\right) \left(\bigotimes_{\ell=1}^{n} e^{-i \gamma_j a_\ell \texttt{Z}_\ell}\right) \right) \bigotimes_{\ell=1}^n \ket{+}_\ell \\
\nonumber    &= \left(\prod_{j=p}^1 \left(\bigotimes_{\ell=1}^{n} \texttt{RX}_\ell(2 \beta_j)\right) \left(\bigotimes_{\ell=1}^{n} \texttt{RZ}_\ell (2 \gamma_j a_\ell)\right) \right) \bigotimes_{\ell=1}^n \ket{+}_\ell ,
\end{align}
where we used the fact that $\texttt{H}\ket{0}=\ket{+}$ in the first step,
the linearity of $H_P$ and $H_M$ allowed us to use the tensor product in the second step, and we expressed the exponentials as rotations in $X$ and $Z$ in the last step. Now we can use the fact that $(A \otimes B)(C \otimes D)=(A C) \otimes (B D)$ to express $\ket{\psi(\gamma, \beta)}$ as:
\begin{align}
    \ket{\psi(\gamma, \beta)} &= \bigotimes_{\ell=1}^{n} \left(\prod_{j=p}^1 \left( \texttt{RX}_\ell(2 \beta_j) \; \texttt{RZ}_\ell (2 \gamma_j a_\ell) \right) \ket{+}_\ell \right) .
\end{align}

Let us now introduce the assumption that all the coefficients in the linear function are positive ($a_\ell>0$). Then, the solution for the spin variables is $s_\ell = 1$, and, in the computational basis, the global optima is state $\ket{0\ldots 0}$. The probability of getting the optimal solution after the measurement is:
\begin{align}
    \nonumber Pr_{opt}(a, \gamma, &\beta) = \left|\braket {0\ldots 0 | \psi(\gamma, \beta)}\right|^2 \\
    \nonumber &= \left| \left(\bigotimes_{\ell=1}^{n} \bra{0}_l \right) \bigotimes_{\ell=1}^{n} \left(\prod_{j=p}^1 \left( \texttt{RX}_\ell(2 \beta_j) \; \texttt{RZ}_\ell (2 \gamma_j a_\ell) \right) \ket{+}_\ell \right) \right|^2 \\
    \nonumber &= \left| \bigotimes_{\ell=1}^{n} \left( \bra{0}_l\prod_{j=p}^1 \left( \texttt{RX}_\ell(2 \beta_j) \; \texttt{RZ}_\ell (2 \gamma_j a_\ell) \right) \ket{+}_\ell \right) \right|^2 \\
\label{eqn:prob_opt}
     &= \prod_{\ell=1}^{n} \left|  \bra{0}\prod_{j=p}^1 \left( \texttt{RX}(2 \beta_j) \; \texttt{RZ} (2 \gamma_j a_\ell) \right) \ket{+}  \right|^2,
\end{align}
where in the last equality we removed the subindices to the states and operators because the modulus is independent of the particular qubits used in the computation.

%\francis{Here we provide the main result. We should explain that we cponsider without loss of generality that all the coefficients are positive. We can provide why this is WLOG when we have the equation.}

\begin{theorem}[Exponential runtime of QAOA with constant $p$ for linear functions]
\label{thm:exponential-qaoa-fixed-p}
    Let the number of layers in the ansatz of QAOA, $p$, be a constant independent of the problem size $n$; and let $a$ be a vector of $m$ positive real values representing a linear Ising model as expressed by Equation~\eqref{eqn:linear-ising}. If $Pr_{opt}(a, \gamma, \beta) < 1$ for all the possible vectors $\gamma$ and $\beta$ of size $p$, then we can build a family of linear Ising models for which the QAOA runtime grows exponentially with the problem size $n$.
\end{theorem}
\begin{proof}
We can build the family of linear functions simply using vector $a$ (of length $m$) repeatedly for the coefficients of the new variables. Let us build a new linear Ising model based on vector $a$ of size $n=k \; m$ variables. In the new Ising model, variables $s_1$, $s_{1+m}$, $s_{1+2m}$, up to $s_{1+(k-1)m}$ all are weighted by coefficient $a_1$. In general, the variables in the set $\{s_{q+\ell m} | 0 \leq \ell < k\}$ are weighted by coefficient $a_q$. With a bit of abuse of notation, we denote the vector of coefficients of the new linear Ising model as $a^k$, which recalls that the coefficients are just $k$ copies of the vector $a$ concatenated.

We assume, as justified previously, that we have an oracle providing the best possible values for $\gamma$ and $\beta$, so we do not need to optimize these vectors using a classical algorithm, and we only need to sample the ansatz until we find an optimal solution. Thus, we compute the probability of obtaining the global optima for the linear Ising model characterized by the vector of coefficients $a^k$. Using Equation~\eqref{eqn:prob_opt} we have:
\begin{align}
\nonumber
    Pr_{opt}(a^k, \gamma, \beta) &= \prod_{\ell=1}^{n} \left|  \bra{0}\prod_{j=p}^1 \left( \texttt{RX}(2 \beta_j) \; \texttt{RZ} (2 \gamma_j a^k_\ell) \right) \ket{+}  \right|^2 \\
\nonumber    
    &= \prod_{\ell=0}^{k-1} \prod_{q=1}^{m} \left|  \bra{0}\prod_{j=p}^1 \left( \texttt{RX}(2 \beta_j) \; \texttt{RZ} (2 \gamma_j a^k_{q+\ell m}) \right) \ket{+}  \right|^2 \\
\nonumber    
    &= \prod_{\ell=0}^{k-1} \prod_{q=1}^{m} \left|  \bra{0}\prod_{j=p}^1 \left( \texttt{RX}(2 \beta_j) \; \texttt{RZ} (2 \gamma_j a_{q}) \right) \ket{+}  \right|^2 \\
\nonumber    
    &= \prod_{\ell=0}^{k-1} Pr_{opt}(a, \gamma, \beta) \\
    &= \left(Pr_{opt}(a, \gamma, \beta)\right)^k,
\end{align}
where we used the fact that $a^k_{q+\ell m} = a_q$ in the second step.

Since, by assumption, we have $Pr_{opt}(a, \gamma, \beta)<1$, the probability of finding the global optimal solution $Pr_{opt}(a^k, \gamma, \beta)$ decreases exponentially with $k=n/m$. Using basic results of the geometric probability distribution,  the average number of times the ansatz has to be measured to find the global optimal solution is:
\begin{equation}
\label{eqn:runtime}
    T(a^k) = \left(\frac{1}{\sqrt[m]{Pr_{opt}(a, \gamma, \beta)}}\right)^{n},
\end{equation}
where the expression within the parentheses is larger than 1 by hypothesis. This means the runtime of QAOA is exponential in the problem size $n$.
\end{proof}
    
Theorem~\ref{thm:exponential-qaoa-fixed-p} contains a big assumption that is at the core of the exponential runtime result. The assumption is that there is a vector $a$ of coefficients for a linear Ising model such that $Pr_{opt}(a, \gamma, \beta) < 1$ for all $\gamma$ and $\beta$ of length $p$. Can we find such a vector given a fixed~$p$? We do it here for small values of $p$. Subsection~\ref{subsec:empirical} will show the results of an experiment that suggests that the sequence of initial consecutive integers $a=(p+1)$, that is, $1, 2, \ldots, p+1$, forms such a vector. We will prove here that this is the case for $p=1$ and $p=2$. We do not have, at the moment, a proof for $p > 2$, only the results shown in Subsection~\ref{subsec:empirical}.

% Let us develop the expression for $Pr_{opt}(a, \gamma, \beta)$ when $a_\ell=\ell$. We will denote this sequence as $(n)$ in the following\footnote{We use parentheses for the sequence instead of square brackets to avoid the confusion with the \emph{set} $\{1, 2, \ldots, n\}$.}. We will first expand the product $RX(2 \beta) \; RZ (2 \gamma)$:

% \begin{align}
% \nonumber    RXZ(2\beta, 2\gamma) &= RX(2 \beta) \; RZ (2 \gamma) \\
% \nonumber &= \rx{\cos \beta}{-i \sin \beta} \rz{e^{-i \gamma}}{e^{i \gamma}} \\
% \nonumber    &=  \matrixtwo{e^{-i \gamma} \cos \beta}{-i e^{i \gamma} \sin \beta}{-i e^{-i \gamma} \sin \beta}{e^{i \gamma} \cos \beta} \\
%     &=  \matrixtwo{(\cos \gamma - i \sin \gamma) \cos \beta}{(-i \cos \gamma + \sin \gamma) \sin \beta}{ (-i \cos \gamma + \sin \gamma) \sin \beta}{(\cos \gamma + i \sin \gamma) \cos \beta} .
% \end{align}

% When $p=1$ and $a_\ell=\ell$ we have:
% \begin{align}
%     \nonumber Pr((2),(\beta_1),(\gamma_1)) &= \prod_{\ell=1}^{2}\left| \bra{0} RXZ(2\beta_1, 2\gamma_1 \ell)\ket{+}\right|^2 \\
%     \nonumber  &= \prod_{\ell=1}^{2}\left| \left( \begin{array}{cc} 1 & 0 \end{array} \right) RXZ(2\beta_1, 2\gamma_1 \ell) \left(\begin{array}{c} 1/\sqrt{2} \\ 1/\sqrt{2} \end{array} \right) \right|^2 \\
% \nonumber    &= \prod_{\ell=1}^{2} \frac{1}{2}\left| ( \cos \gamma_1 \ell - i \sin \gamma_1 \ell) \cos \beta_1 + (-i \cos \gamma_1 \ell + \sin \gamma_1 \ell) \sin \beta_1 \right|^2 \\
% &= \prod_{\ell=1}^{2} \frac{1}{2}\left| ( \cos (\gamma_1 \ell - \beta_1) - i \sin (\gamma_1 \ell+\beta_1) \right|^2
% \end{align}

\begin{theorem}
\label{thm:qaoa-p1-linear}
    The probability of measuring the optimal solution of the linear Ising model with coefficients $a=(2)=(1,2)$ using the QAOA ansatz with $p=1$ holds $Pr(a, \gamma, \beta) < 1$ for all $\gamma$ and $\beta$.
\end{theorem}
\begin{proof}
    Observe that $\gamma$ and $\beta$ are vectors of size 1 in this case and they could be interpreted as scalars, but we will use the notation $\gamma_1$ to refer to the only element of vector $\gamma$ in the following. 
    
    We prove the claim by contradiction. Assume that $Pr(a, \gamma, \beta) = 1$.  As a consequence, the probability of measuring 0 in each of the qubits must also be 1:
    \begin{eqnarray}
    \label{eqn:0-psi}
        && \left| \braket{0 |\psi(\gamma, \beta)} \right|^2 = \left| \bra{0} \texttt{RX}(2 \beta_1) \; \texttt{RZ} (2 \gamma_1) \ket{+}\right|^2 = 1, \\
        && \left| \braket{0 |\psi(2\gamma, \beta)} \right|^2 = \left| \bra{0} \texttt{RX}(2 \beta_1) \; \texttt{RZ} (4 \gamma_1) \ket{+}\right|^2 = 1 .
    \end{eqnarray}

    Since all the quantum states are normalized, this implies that the states $\ket{\psi(\gamma,\beta)}$ and $\ket{\psi(2\gamma,\beta)}$ are the same except for a possible phase factor:  $\ket{\psi(2 \gamma,\beta)} = e^{i\alpha} \ket{\psi(\gamma,\beta)}$. If we multiply both sides by $\bra{\psi(\gamma,\beta)}$ we can write $\braket{\psi(\gamma,\beta)|\psi(2\gamma,\beta)} = e^{i\alpha}$. Developing this expression we have:
    \begin{align}
    \nonumber    \braket{\psi(\gamma,\beta)|\psi(2\gamma,\beta)} &= 
            \bra{+} \texttt{RZ} (-2 \gamma_1)  \texttt{RX}(-2 \beta_1) \texttt{RX}(2 \beta_1)  \texttt{RZ} (4 \gamma_1) \ket{+} \\
    \nonumber        &= \bra{+} \texttt{RZ} (-2 \gamma_1) \; \texttt{RZ} (4 \gamma_1) \ket{+} \\
    \nonumber        &= \bra{+} \texttt{RZ} (2 \gamma_1) \ket{+} \\
    \nonumber        &= \frac{1}{2}\left(\begin{array}{cc}1 & 1 \end{array}\right) \rz{e^{-i \gamma_1}}{e^{i \gamma_1}}\left(\begin{array}{c}1 \\ 1\end{array}\right) \\
                     &= \frac{e^{-i \gamma_1}+e^{i \gamma_1}}{2} = \cos \gamma_1 .
    \end{align}

    The previous expression has modulus one only if $\gamma_1 = k \pi$ for some integer $k$. Now let us compute the modulus of $\braket{0 | \psi((k \pi), \beta)}$:
    \begin{align}
\nonumber        \left|\braket{0 | \psi((k \pi), \beta)} \right| &= \left|\bra{0} \texttt{RX}(2 \beta_1) \texttt{RZ}(2 k \pi) \ket{+} \right|\\
\nonumber            &= \left|(-1)^{k} \bra{0} \texttt{RX}(2 \beta_1) \ket{+} \right| \\
\nonumber            &= \left|e^{-i \beta_1} \braket{0|+} \right|\\
            &= \left|\braket{0|+} \right| = \frac{1}{\sqrt{2}} < 1,
    \end{align}
where we used $\texttt{RZ}(2k \pi)=(-1)^k$ and $\texttt{RX}(2 \beta_1)=e^{-i \beta_1}\ket{+}$. This result contradicts Equation~\eqref{eqn:0-psi}.
\end{proof}

We can exactly compute the maximum value of $Pr((1,2), \gamma, \beta)$ for any $\gamma$ and $\beta$. The detailed computation can be found in the supplementary material~\cite{supplementary-qaoa-runtime}. The idea is to express $Pr((1,2), \gamma, \beta)$ as a multivariate polynomial, replacing the trigonometric expressions by variables. Then, with the
help of the theory of ideals of multivariate polynomials and Gr\"obner bases~\cite{10.1007/3-540-12868-9_99} it is possible to find a polynomial that contains the maximum as one of its roots. This polynomial is $5832 x^3 - 6804 x^2 + 1472 x -8$, and its maximum root is
\begin{align}
\nonumber    \max_{\gamma, \beta} Pr((1,2), \gamma, \beta) &= \frac{\sqrt[3]{\frac{1}{2} \left(111628119168+3824485632 \; i \; \sqrt{1518}\right)}}{17496} \\
\nonumber &+\frac{1174}{\sqrt[3]{\frac{1}{2} \left(111628119168+3824485632 \; i \; \sqrt{1518}\right)}} \\
\nonumber &+\frac{7}{18}\\
    &\approx 0.882385 < 1
\end{align}

% A similar result can be proven for $p=2$ and $a=(3)$. \francis{I don't know if I will have the time to develop this theorem. My current proof is based on a geometric argument over the Bloch sphere. But I think an algebraic argument similar to the one in Theorem~\ref{thm:qaoa-p1-linear} should be possible (I simply did not find it yet). I will focus now in the rest of the paper.}

%\francis{We explain here the trick of using a transformation that transform the maximization problem into a system of multivariate polynomials. Then we explain that the Gr\"obner bases provides a polynomial that contains in its roots all the local optima of the problem. Then, we only need to evaluate the polynomial in 1 to check if the probability of finding the optimal solution is 1 or not. }

%\francis{Then, we explain that if the probability is lower than 1, then the runtime is exponential. We can intertwin some examples.}

\begin{theorem}
\label{thm:qaoa-p2-linear}
    The probability of measuring the optimal solution of the linear Ising model with coefficients $a=(3)=(1,2,3)$ using the QAOA ansatz with $p=2$ holds $Pr(a, \gamma, \beta) < 1$ for all $\gamma$ and $\beta$.
\end{theorem}
\begin{proof}
    The proof is solved by contradiction, as in Theorem~\ref{thm:qaoa-p1-linear}. Let us assume that $Pr(a, \gamma, \beta) = 1$.  The probability of measuring 0 in each of the qubits must also be 1:
    \begin{align*}
    %\label{eqn:0-psi-p2}
        \nonumber \left| \braket{0 |\psi(\gamma, \beta)} \right|^2 &= \left| \bra{0} \texttt{RX}(2 \beta_2) \; \texttt{RZ} ( 2 \gamma_2) \; \texttt{RX}(2 \beta_1) \; \texttt{RZ} (2 \gamma_1) \ket{+}\right|^2 = 1, \\
       \nonumber  \left| \braket{0 |\psi(2\gamma, \beta)} \right|^2 &= \left| \bra{0} \texttt{RX}(2 \beta_2) \; \texttt{RZ} (4 \gamma_2) \; \texttt{RX}(2 \beta_1) \; \texttt{RZ} (4 \gamma_1) \ket{+}\right|^2 = 1 , \\
       \nonumber  \left| \braket{0 |\psi(3\gamma, \beta)} \right|^2 &= \left| \bra{0} \texttt{RX}(2 \beta_2) \; \texttt{RZ} (6 \gamma_2) \; \texttt{RX}(2 \beta_1) \; \texttt{RZ} (6 \gamma_1) \ket{+}\right|^2 = 1 .
    \end{align*}
    We will focus now in states $\ket{\psi(\gamma, \beta)}$ and $\ket{\psi(2\gamma, \beta)}$ and will find values for $\gamma$ and $\beta$ that make them simultaneously equal to $\ket{0}$, except for an irrelevant global phase. The proportionality between the two states implies that the complex value
    \begin{align*}
        &\braket{\psi(\gamma,\beta)|\psi(2\gamma,\beta)} \\
        &\phantom{abcd}= \bra{+} \texttt{RZ} (-2 \gamma_1)  \texttt{RX}(-2 \beta_1) \texttt{RZ} (2 \gamma_2) \texttt{RX}(2 \beta_1)  \texttt{RZ} (4 \gamma_1) \ket{+} \\
        &\phantom{abcd}= \braket{\psi((\gamma_1),(\beta_1))|\texttt{RZ}(2 \gamma_2)|\psi((2\gamma_1),(\beta_1))}
    \end{align*}
    must have modulus one. Since the $\psi$ states have modulus one and $\texttt{RZ}$ is a unitary operator, by Schwarz-Cauchy inequality, we have that the previous expression has modulus one if and only if
    \begin{align}
      \nonumber  \left| \braket{0|\psi((\gamma_1),(\beta_1))} \right| &= \left| \braket{0|\texttt{RZ}(2 \gamma_2) |\psi((2\gamma_1),(\beta_1))} \right|  \\
      \label{eqn:psi1-2-eq}  &= \left| \braket{0|\psi((2\gamma_1),(\beta_1))} \right| ,
    \end{align}
    where in the last step we used the fact that $\bra{0}\texttt{RZ}(2\gamma_2) = e^{i \gamma_2}\bra{0}$. If we develop Equation~\eqref{eqn:psi1-2-eq} with the help of the matrix definition of the operator we find the following constraint for $\gamma_1$ and $\beta_1$:
    \begin{equation}
    \label{eqn:sins}
        \sin 2 \beta_1 \sin 2 \gamma_1 = \sin 2 \beta_1 \sin 4 \gamma_1 = 2 \sin 2 \beta_1 \sin 2\gamma_1 \cos 2 \gamma_1.
    \end{equation}
    Two solutions to the previous equation are $\beta_1=0$ and $\beta_1=\pi/2$. In the first case, $\texttt{RX}(2\beta_1)$ is the identity and $\ket{\psi(\gamma, \beta)}$ is composed of one rotation around the $Z$ axis followed by one rotation around the $X$ axis. According to Theorem~\ref{thm:qaoa-p1-linear}, these two rotations cannot make $Pr_{opt}((1,2),\gamma,\beta)=1$ and, thus, they will not make $Pr_{opt}((1,2,3),\gamma,\beta)=1$. If $\beta_1=\pi/2$ it is also easy to prove that $\ket{\psi(\gamma, \beta)}$ contains again one rotation around the $Z$ axis followed by one rotation around the $X$ axis. Thus, we can assume that $\sin 2\beta_1 \neq 0$ and we can divide Equation~\eqref{eqn:sins} by $\sin 2\beta_1$ to obtain:
    \begin{equation}
    \label{eqn:sins2}
        \sin 2 \gamma_1 = 2 \sin 2\gamma_1 \cos 2 \gamma_1 .
    \end{equation}
    We can see that if $\gamma_1=0$ or $\gamma_1=\pi/2$ the equality holds. But $\texttt{RX}(2\beta_1)\texttt{RZ}(0)\ket{+}=e^{-i \beta_1}\ket{+}$ and the state $\ket{\psi(\gamma, \beta)}$ again reduces to a rotation in $Z$ followed by one rotation in $X$. The same happens if $\gamma_1=\pi/2$: $\texttt{RX}(2\beta_1)\texttt{RZ}(\pi/2)\ket{+}=e^{i \beta_1}\ket{-}$. Thus, we can assume that $\sin 2\gamma_1 \neq 0$ and we can divide Equation~\eqref{eqn:sins2} by $\sin 2\gamma_1$. The resulting expression is $\cos 2\gamma_1 = 1/2$ and the possible solutions are $\gamma_1 = \pm \pi/6$.

    Now let's introduce the third qubit. We know the state of the third qubit must also fulfill
    \[|\braket{\psi(\gamma,\beta)|\psi(3\gamma,\beta)}|=1\]
    because $\ket{\psi(3\gamma,\beta)}$ needs to be proportional to $\ket{\psi(\gamma,\beta)}$ except for an irrelevant global phase. Following the same steps done before, we find the constraint 
    \begin{equation}
        \sin 2 \gamma_1 = \sin 6 \gamma_1 ,
    \end{equation}
    which is not satisfied by $\gamma_1=\pm \pi/6$. This completes the proof.
\end{proof}

We are now ready to announce the following conjecture.

\begin{conjecture}
\label{conjecture}
    There is a linear Ising model with coefficients' vector $a$ containing $p+1$ different numbers such that the probability of measuring the optimal solution of the Ising model using the QAOA ansatz with $p$ layers holds the inequality $Pr(a, \gamma, \beta) < 1$ for all $\gamma$ and $\beta$.
\end{conjecture}

\subsection{Experimental evidences}
\label{subsec:empirical}

We can think of an intuitive justification for Conjecture~\ref{conjecture}: it seems impossible to adjust $p$ rotations around the $Z$ axis and $p$ rotations around the $X$ axis of the Bloch sphere to make the $p+1$ independent qubits to reach state $\ket{0}$ from state $\ket{+}$. We have observed for low values of $p$ that with $2p$ rotation angles we can reach state $\ket{0}$ for all qubits from state $\ket{+}$, but the number of ways in which this can be achieved is finite. We think this finite set of solutions for the first $p$ qubits will not work with qubit $p+1$.
However, at the time of writing these lines, we still do not have a general proof for $p > 2$. 

In order to support Conjecture~\ref{conjecture}, we did some experiments. In particular, we consider $Pr((m),\gamma,\beta)$ as a function of the vectors $\gamma$ and $\beta$, both of length $p$, and we maximized $Pr((m),\gamma,\beta)$ using approximated methods. The methods used were Nelder-Mead, Differential Evolution, Simulated Annealing, and Random Search, all of them available in the \texttt{NMaximize} function of the Wolfram software package (version 14.2).\footnote{The code can be found in the supplementary material~\cite{supplementary-qaoa-runtime}.} We applied the maximization for $p\in[5]$ and $m\in[7]$, and for each combination of $m$ and $p$, we selected the maximum value obtained by any of these four methods. The results obtained are shown in Table~\ref{tab:probs}. We can observe that for $m > p$ the probabilities obtained are below 1, except for the case $m=5$ and $p=4$. In this latter case, we cannot discard rounding errors, but it is also important to recall that we are assuming that the vector of coefficients $(m)$, composed of consecutive integer numbers, always works. This assumption does not need to be true in general. Perhaps some particular values of $p$ need a different combination of coefficients. We defer this analysis to future research.

In Table~\ref{tab:probs}, we also observe the exact value we obtained for $p=1$ and $m=2$  (0.882385).
Another interesting observation is that even in the cases where $Pr((m),\gamma,\beta) < 1$ for $m=p+1$, the maximum probability is approaching 1 as $p$ increases. Thus, the base of the exponent in the runtime approaches 1 as $p$ is larger, and we need to make $m$ larger than $p+1$ to have appreciably large numbers in the base of the exponential runtime. Using Equation~\eqref{eqn:runtime} we can compute the base of the exponent for the combinations in Table~\ref{tab:probs}. The results are in Table~\ref{tab:exponent-base}. While the base of the exponent is relatively close to 1 in most of the cases of Table~\ref{tab:exponent-base}, the values increase with $m$, and it is still not clear what the upper bound is (if there is one).

%\francis{We show here some values for the polynomial for different values of $n$ and $p$ in a table. We can show the maximum we can. Then we formulate the conjecture.}

\begin{table}[!ht]
\centering
\caption{Approximated maximum probability of measuring the global optimum for the vector of coefficients $(m)$ in the QAOA variational circuit with $p$ layers.}
\label{tab:probs}
\begin{tabular}{rrrrrr}
\toprule
$m$ & \multicolumn{1}{c}{$p=1$} & \multicolumn{1}{c}{$p=2$} & \multicolumn{1}{c}{$p=3$} & \multicolumn{1}{c}{$p=4$} & \multicolumn{1}{c}{$p=5$} \\
\midrule
1 & 1.000000 & 1.000000 & 1.000000 & 1.000000 & 1.000000 \\
2 & 0.882385 & 1.000000 & 1.000000 & 1.000000 & 1.000000 \\
3 & 0.761904 & 0.944816 & 1.000000 & 1.000000 & 1.000000 \\
4 & 0.652920 & 0.881947 & 0.997275 & 1.000000 & 1.000000 \\
5 & 0.557571 & 0.817512 & 0.974482 & 1.000000 & 1.000000 \\
6 & 0.475241 & 0.754870 & 0.948152 & 0.999680 & 0.999995 \\
7 & 0.404604 & 0.695386 & 0.921055 & 0.965912 & 0.999675 \\    
\bottomrule
\end{tabular}
\end{table}

\begin{table}[!ht]
\centering
\caption{Approximated base of the exponential runtime of an optimally configured QAOA solving linear functions for the vector of coefficients $(m)$ and number of layers $p$.}
\label{tab:exponent-base}
\begin{tabular}{rrrrrr}
\toprule
$m$ & \multicolumn{1}{c}{$p=1$} & \multicolumn{1}{c}{$p=2$} & \multicolumn{1}{c}{$p=3$} & \multicolumn{1}{c}{$p=4$} & \multicolumn{1}{c}{$p=5$} \\
\midrule
1 & 1.00000 & 1.00000 & 1.00000 & 1.00000 & 1.00000 \\
2 & 1.06456 & 1.00000 & 1.00000 & 1.00000 & 1.00000 \\
3 & 1.09488 & 1.01910 & 1.00000 & 1.00000 & 1.00000 \\
4 & 1.11246 & 1.03190 & 1.00068 & 1.00000 & 1.00000 \\
5 & 1.12393 & 1.04112 & 1.00518 & 1.00000 & 1.00000 \\
6 & 1.13200 & 1.04798 & 1.00891 & 1.00005 & 1.00000 \\
7 & 1.13799 & 1.05327 & 1.01182 & 1.00497 & 1.00005 \\
\bottomrule
\end{tabular}
\end{table}

\section{Discussion \label{sec:discussion}}

Quantum computing is expected to provide supremacy over classical computation for some problems~\cite{eickbusch2024demonstratingdynamicsurfacecodes}. QAOA, in particular, has shown supremacy in generating probability distributions that are NP-hard to generate in a classical computer~\cite{ref9}. However, we are interested in optimization and how quantum computers can help optimize functions more efficiently than classical computers. In this sense, the results in the previous section show that QAOA \emph{may not} be the best algorithm for optimization, showing exponential runtime for optimizing some linear functions. Even simple black-box randomized search heuristics running in classical computers have been proven to find the optimal solution to linear functions in polynomial time~\cite{NeumannWitt2022}. In general, linear functions can be solved in $O(n)$ in a classical computer, as Equation~\eqref{eqn:solution-linear} suggests.

We should clarify here that we assume throughout the paper that we can get perfect quantum computers with no noise, and this is the context in which our results hold. We are, however, in the NISQ era, with noisy circuit implementations. Could QAOA be a suitable algorithm for linear functions in the NISQ era? Intuitively, we can think that the effect of noise on the QAOA implementation will ``blur'' the probability distributions obtained by QAOA. It is difficult to believe, in the authors' opinion, that this distortion of the probability distribution in the measurement can have as an effect the concentration of probability on the state representing the optimal solution $\ket{0\ldots 0}$. We need further investigation to answer this question rigorously.

Observe that we say QAOA \emph{may not} be the best algorithm because we only considered linear functions. We cannot discard the possibility that there is a particular family of NP-hard problems for which QAOA can find optimal solutions more efficiently than classical computers.  In a recent paper, Boulebnane and Montanaro~\cite{Boulebnane2024} suggest that QAOA may outperform classical algorithms for $k$-SAT. However, there have been other cases in which a new result has refuted an apparently quantum advantage of QAOA. One example is the work by Barak et al.~\cite{DBLP:journals/corr/BarakMORRSTVWW15} focused on the Max-$k$XOR constraint satisfaction problem. Quoting Boulebnane and Montanaro, ``although there have been many [...] works on the theoretical and empirical performance of QAOA for optimization problems [...], none has yet shown an unambiguous advantage over the best classical algorithms.''
Future work should investigate this possibility. However, if Conjecture~\ref{conjecture}  is true, we can discard the possibility that QAOA achieves subexponential runtimes for the family of NP-hard problems, since linear functions are included in this family. In this sense, this result aligns with the one of Harrow and Napp~\cite{PhysRevLett.126.140502} that show that \emph{any} hybrid classical-quantum algorithm in a black-box setting, that is, based only on the evaluation of the objective function, requires $\Omega(n^3/\varepsilon^2)$ calls to the oracle computing the objective function, where $n$ is the number of qubits and $\varepsilon$ is the maximum tolerated error to the optimal solution. This means that as we approach the optimal solution, the runtime increases.

Does it mean that a quantum computer cannot solve linear functions efficiently? No, it is very easy to design a quantum circuit with $n$ gates computing the optimal solution to a linear function. Fig.~\ref{fig:solving-linear} shows one of these circuits. Thus, we need better quantum optimization algorithms for gate-based quantum computers. We can try to find such algorithms by looking at the quantum algorithms with proved improvements over the classical methods. One such algorithm is Grover's search, which is also a black-box algorithm (it does not consider the details of the problem, but only its evaluation function). An optimization problem can be transformed into a series of decision problems using dichotomic search. Then, Grover’s search could solve each decision problem, which has a quadratic speedup compared to the classical search algorithm. However, the original Grover's algorithm assumes that the search space is the set of all binary strings of length $n$, which size is $O(2^n)$. Then, Grover's search runtime is $O(1.4142^n)$. We have classical algorithms that explore the $O(2^n)$ search space more efficiently for some particular problems, like SAT or TSP.

\begin{figure}[!ht]
\centering
\begin{quantikz}
\lstick[2]{$a_1$} &  & \wire[d][1][dotted]{a}  &  & & & \\
 &  \ctrl{6} &  &  &  & &\\
\lstick[2]{$a_2$} &  &  & \wire[d][1][dotted]{a} &  & &\\
 &  & \ctrl{5} &  &  & &\\
 \wave &&&&&& \\
\lstick[2]{$a_n$} &  &  &  & \wire[d][1][dotted]{a} & & \\
 &  &  & \ctrl{4} &  & &\\ 
\lstick{$\ket{0}$} & \targ{} &  & &  &\meter{} & \rstick[4]{output}\\
\lstick{$\ket{0}$} &  & \targ{} &  & &\meter {} &\\
 \wave &&&&&& \\
\lstick{$\ket{0}$} &  &  & \targ{} & &\meter{} &
\end{quantikz}
%\Description[Quantum circuit computing the optimal solution to a linear Ising model.]{Quantum circuit computing the optimal solution to a linear Ising model.}
    \caption{Quantum circuit computing the optimal solution to a linear Ising model. The first $n$ registers contain the coefficients of the linear Ising model in two's complement. The most significant bit of each register is the one at the bottom and represents the sign of the integer number. The sign bit of each register controls one of the output qubits.}
    \label{fig:solving-linear}
\end{figure}
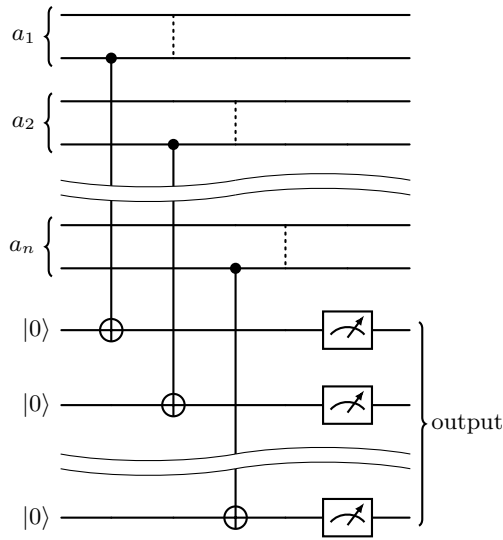

In the last decades, new quantum algorithms have been published that consider these classical algorithmic improvements and provide a quadratic speedup compared to the corresponding classical algorithm. One example is the framework proposed by Ashley Montanaro for the family of constraint satisfaction problems~\cite{Montanaro2019}.

Unfortunately, quadratic speedup is insufficient to claim an advantage of quantum computers over classical ones.
In a recent article in \textit{Communications of the ACM}, Hoefler et al.~\cite{HoeflerHT23} highlight that for noiseless gate-based quantum computers to compete with current GPUs, the acceleration obtained in algorithms must be more than quadratic. This acceleration should be cubic, quartic, or higher. The work compares current GPUs against ideal noiseless quantum computers (unavailable today).
However, Bennett et al.~\cite{BennettBBV97} proved that the quadratic speedup achieved by Grover’s search is optimal in a black box setting. Moreover, it has been proven that solving any more general predicate on a quantum computer that requires querying an oracle cannot achieve a speedup greater than polynomial~\cite{BealsBCMW98}. In particular, it cannot be faster than the sixth root of the execution time on a classical computer~\cite[p. 276]{ref1}. This leaves room to surpass current parallel processing, but it does not establish quantum supremacy as clearly as the exponential speedup achieved with the quantum Fourier transform. To achieve more than polynomial acceleration, it is necessary to \emph{open the box} and abandon black-box optimization.

In the search for the building blocks of new quantum algorithms for optimization, it is natural to ask if we can take advantage of the algorithmic approaches that already showed an exponential speedup, like the mentioned quantum Fourier transform. A first step could be finding ways to use the Fourier transform in optimization with classical algorithms. One example is the algorithm proposed by Risi Kondor to optimize the Quadratic Assignment Problem using its group-theoretic Fourier transform~\cite{Kondor10}. Another example is the exponential speedup obtained by the Partition Crossover operator~\cite{TintosWC15} during recombination in evolutionary algorithms and iterated local search. More recently, it has been proposed that the group-theoretic Fourier transform can be at the core of the design of efficient operators used by classical algorithms during the search~\cite{DBLP:conf/ppsn/ChicanoWOT24}. Interestingly, at least one group-theoretic Fourier transform has an exponential speedup in gate-based quantum computers: the Wash-Hadamard transform. This transform requires exponential time in a classical computer and it can be implemented in a gate-based quantum computer using just $n$ Hadamard gates (one per qubit). The Deutsch-Jozsa~\cite{DeutschJozsa92} and Bernstein-Vazirani~\cite{doi:10.1137/S0097539796300921} algorithms leverage the Walsh-Hadamard transform to get an exponential speedup. The celebrated and feared Shor's algorithm~\cite{doi:10.1137/S0097539795293172} is also based on a Quantum Fourier Transform. Thus, a Fourier transform could guide the design of new promising quantum optimization algorithms in the context of noiseless gate-based quantum computers.

\section{Conclusions}
\label{sec:conclusions}

We have seen in this paper that QAOA has trouble finding the global optimal solution for linear functions, requiring exponential time if the number of layers $p$ is fixed. This result raises new questions and opens new promising future lines of research. While we have provided the proofs of Conjecture~\ref{conjecture} for $p=1$ and $p=2$, we still miss a proof for $p>2$. We also wonder if the base of the exponential runtime is upper-bounded. It would be interesting to know if QAOA can provide a clear advantage over classical algorithms for some particular families of NP-hard problems, even if the runtime is exponential. Observe that there is no polynomial-time classical algorithm to solve NP-hard problems, and in the field of exponential algorithms, reducing the base of the exponential is an advantage~\cite{10.5555/1941886}.

The discussion of Section~\ref{sec:discussion} suggests that any problem that can be efficiently solved with the help of a Fourier transform could be translated into a quantum algorithm with a clear advantage over classical algorithms because noiseless gate-based quantum computers can compute different kinds of Fourier transforms in polynomial time. Future research should focus on finding practical ways of using the Fourier transform to solve optimization problems and designing quantum algorithms to compute the general group-theoretic Fourier transform over finite groups.

\section*{Acknowledgements}
This research is sponsored by Universidad de Málaga under grant PAR~4/2023. Grammarly, Writefull, and ChatGPT have been used during the writing process to revise the grammar.

%\bibliographystyle{splncs04}
%\bibliography{bibliography.bib}

\begin{thebibliography}{10}
\providecommand{\url}[1]{\texttt{#1}}
\providecommand{\urlprefix}{URL }
\providecommand{\doi}[1]{https://doi.org/#1}

\bibitem{google2019}
Arute, F., Arya, K., Babbush, R., Bacon, D., Bardin, J.C., Barends, R., Biswas,
  R., Boixo, S., Brandao, F.G., Buell, D.A., Burkett, B., Chen, Y., Chen, Z.,
  Chiaro, B., Collins, R., Courtney, W., Dunsworth, A., Farhi, E., Foxen, B.,
  Fowler, A., Gidney, C., Giustina, M., Graff, R., Guerin, K., Habegger, S.,
  Harrigan, M.P., Hartmann, M.J., Ho, A., Hoffmann, M., Huang, T., Humble,
  T.S., Isakov, S.V., Jeffrey, E., Jiang, Z., Kafri, D., Kechedzhi, K., Kelly,
  J., Klimov, P.V., Knysh, S., Korotkov, A., Kostritsa, F., Landhuis, D.,
  Lindmark, M., Lucero, E., Lyakh, D., Mandrà, S., McClean, J.R., McEwen, M.,
  Megrant, A., Mi, X., Michielsen, K., Mohseni, M., Mutus, J., Naaman, O.,
  Neeley, M., Neill, C., Niu, M.Y., Ostby, E., Petukhov, A., Platt, J.C.,
  Quintana, C., Rieffel, E.G., Roushan, P., Rubin, N.C., Sank, D., Satzinger,
  K.J., Smelyanskiy, V., Sung, K.J., Trevithick, M.D., Vainsencher, A.,
  Villalonga, B., White, T., Yao, Z.J., Yeh, P., Zalcman, A., Neven, H.,
  Martinis, J.M.: Quantum supremacy using a programmable superconducting
  processor. Nature  \textbf{574},  505--510 (2019).
  \doi{10.1038/s41586-019-1666-5}

\bibitem{DBLP:journals/corr/BarakMORRSTVWW15}
Barak, B., Moitra, A., O'Donnell, R., Raghavendra, P., Regev, O., Steurer, D.,
  Trevisan, L., Vijayaraghavan, A., Witmer, D., Wright, J.: Beating the random
  assignment on constraint satisfaction problems of bounded degree. CoRR
  \textbf{abs/1505.03424} (2015), \url{http://arxiv.org/abs/1505.03424}

\bibitem{BealsBCMW98}
Beals, R., Buhrman, H., Cleve, R., Mosca, M., de~Wolf, R.: Quantum lower bounds
  by polynomials. In: 39th Annual Symposium on Foundations of Computer Science,
  {FOCS} '98, November 8-11, 1998, Palo Alto, California, {USA}. pp. 352--361.
  {IEEE} Computer Society (1998). \doi{10.1109/SFCS.1998.743485},
  \url{https://doi.org/10.1109/SFCS.1998.743485}

\bibitem{BennettBBV97}
Bennett, C.H., Bernstein, E., Brassard, G., Vazirani, U.V.: Strengths and
  weaknesses of quantum computing. {SIAM} J. Comput.  \textbf{26}(5),
  1510--1523 (1997). \doi{10.1137/S0097539796300933},
  \url{https://doi.org/10.1137/S0097539796300933}

\bibitem{doi:10.1137/S0097539796300921}
Bernstein, E., Vazirani, U.: Quantum complexity theory. SIAM Journal on
  Computing  \textbf{26}(5),  1411--1473 (1997).
  \doi{10.1137/S0097539796300921},
  \url{https://doi.org/10.1137/S0097539796300921}

\bibitem{ref8}
Blekos, K., Brand, D., Ceschini, A., Chou, C.H., Li, R.H., Pandya, K., Summer,
  A.: A review on quantum approximate optimization algorithm and its variants.
  Physics Reports  \textbf{1068},  1--66 (2024).
  \doi{https://doi.org/10.1016/j.physrep.2024.03.002},
  \url{https://www.sciencedirect.com/science/article/pii/S0370157324001078}

\bibitem{Boulebnane2024}
Boulebnane, S., Montanaro, A.: Solving boolean satisfiability problems with the
  quantum approximate optimization algorithm. PRX Quantum  \textbf{5},  030348
  (9 2024). \doi{10.1103/prxquantum.5.030348}

\bibitem{supplementary-qaoa-runtime}
Chicano, F., Dahi, Z.A., Luque, G.: Supplementary material of the paper
  entitled ``{The Quantum Approximate Optimization Algorithm Can Require
  Exponential Time to Optimize Linear Functions}'' published in the {Quantum
  Optimization Workshop at GECCO 2025} (2025). \doi{10.5281/zenodo.15319951},
  \url{https://doi.org/10.5281/zenodo.15319951},
  https://doi.org/10.5281/zenodo.15319951

\bibitem{DBLP:conf/ppsn/ChicanoWOT24}
Chicano, F., Whitley, D., Ochoa, G., Tin{\'{o}}s, R.: Generalizing and unifying
  gray-box combinatorial optimization operators. In: Affenzeller, M., Winkler,
  S.M., Kononova, A.V., Trautmann, H., Tusar, T., Machado, P., B{\"{a}}ck, T.
  (eds.) Parallel Problem Solving from Nature - {PPSN} {XVIII} - 18th
  International Conference, {PPSN} 2024, Hagenberg, Austria, September 14-18,
  2024, Proceedings, Part {I}. Lecture Notes in Computer Science, vol. 15148,
  pp. 52--67. Springer (2024). \doi{10.1007/978-3-031-70055-2\_4},
  \url{https://doi.org/10.1007/978-3-031-70055-2\_4}

\bibitem{DeutschJozsa92}
Deutsch, D., Jozsa, R.: Rapid solution of problems by quantum computation.
  Proceedings of the Royal Society A  \textbf{439}(1907),  553--558 (1992).
  \doi{https://doi.org/10.1098/rspa.1992.0167}

\bibitem{eickbusch2024demonstratingdynamicsurfacecodes}
Eickbusch, A., McEwen, M., Sivak, V., Bourassa, A., Atalaya, J., Claes, J.,
  Kafri, D., Gidney, C., Warren, C.W., Gross, J., Opremcak, A., Miao, N.Z.K.C.,
  Roberts, G., Satzinger, K.J., Bengtsson, A., Neeley, M., Livingston, W.P.,
  Greene, A., Rajeev, Acharya, Beni, L.A., Aigeldinger, G., Alcaraz, R.,
  Andersen, T.I., Ansmann, M., Frank, Arute, Arya, K., Asfaw, A., Babbush, R.,
  Ballard, B., Bardin, J.C., Bilmes, A., Jenna, Bovaird, Bowers, D., Brill, L.,
  Broughton, M., Browne, D.A., Buchea, B., Buckley, B.B., Tim, Burger, Burkett,
  B., Bushnell, N., Cabrera, A., Campero, J., Chang, H.S., Chiaro, B., Chih,
  L.Y., Cleland, A.Y., Cogan, J., Collins, R., Conner, P., Courtney, W.,
  Alexander, Crook, L., Curtin, B., Das, S., Barba, A.D.T., Demura, S.,
  Lorenzo, L.D., Paolo, A.D., Donohoe, P., Drozdov, I.K., Dunsworth, A., Elbag,
  A.M., Elzouka, M., Erickson, C., Ferreira, V.S., Burgos, L.F., Forati, E.,
  Fowler, A.G., Foxen, B., Ganjam, S., Gonzalo, Garcia, Gasca, R., Élie
  Genois, Giang, W., Gilboa, D., Gosula, R., Dau, A.G., Dietrich, Graumann, Ha,
  T., Habegger, S., Hansen, M., Harrigan, M.P., Harrington, S.D., Heslin, S.,
  Heu, P., Higgott, O., Hiltermann, R., Hilton, J., Huang, H.Y., Huff, A.,
  Huggins, W.J., Jeffrey, E., Jiang, Z., Jin, X., Jones, C., Joshi, C., Juhas,
  P., Kabel, A., Kang, H., Amir, Karamlou, H., Kechedzhi, K., Khaire, T.,
  Khattar, T., Khezri, M., Kim, S., Kobrin, B., Korotkov, A.N., Kostritsa, F.,
  Kreikebaum, J.M., Kurilovich, V.D., Landhuis, D., Tiano, Lange-Dei, Langley,
  B.W., Lau, K.M., Ledford, J., Lee, K., Lester, B.J., Guevel, L.L., Wing, Li,
  Y., Lill, A.T., Locharla, A., Lucero, E., Lundahl, D., Lunt, A., Madhuk, S.,
  Maloney, A., Mandrà, S., Martin, L.S., Martin, O., Maxfield, C., McClean,
  J.R., Meeks, S., Anthony, Megrant, Molavi, R., Molina, S., Montazeri, S.,
  Movassagh, R., Newman, M., Nguyen, A., Nguyen, M., Ni, C.H., Oas, L., Orosco,
  R., Ottosson, K., Pizzuto, A., Potter, R., Pritchard, O., Quintana, C.,
  Ramachandran, G., Reagor, M.J., Rhodes, D.M., Rosenberg, E., Rossi, E.,
  Sankaragomathi, K., Schurkus, H.F., Shearn, M.J., Shorter, A., Shutty, N.,
  Shvarts, V., Small, S., Smith, W.C., Springer, S., Sterling, G., Suchard, J.,
  Szasz, A., Sztein, A., Thor, D., Tomita, E., Torres, A., Torunbalci, M.M.,
  Vaishnav, A., Vargas, J., Sergey, Vdovichev, Vidal, G., Heidweiller, C.V.,
  Waltman, S., Waltz, J., Wang, S.X., Ware, B., Weidel, T., White, T., Wong,
  K., Woo, B.W.K., Woodson, M., Xing, C., Yao, Z.J., Yeh, P., Ying, B., Yoo,
  J., Yosri, N., Young, G., Zalcman, A., Yaxing, Zhang, Zhu, N., Boixo, S.,
  Kelly, J., Smelyanskiy, V., Neven, H., Bacon, D., Chen, Z., Klimov, P.V.,
  Roushan, P., Neill, C., Chen, Y., Morvan, A.: Demonstrating dynamic surface
  codes (2024), \url{https://arxiv.org/abs/2412.14360}, arXiv:2412.14360

\bibitem{farhi14-qaoa}
Farhi, E., Goldstone, J., Gutmann, S.: A quantum approximate optimization
  algorithm (2014), \url{https://arxiv.org/abs/1411.4028}, arXiv:1411.4028

\bibitem{ref9}
Farhi, E., Harrow, A.W.: Quantum supremacy through the quantum approximate
  optimization algorithm (2019), \url{https://arxiv.org/abs/1602.07674},
  arXiv:1602.07674

\bibitem{10.5555/1941886}
Fomin, F.V., Kratsch, D.: Exact Exponential Algorithms. Springer-Verlag,
  Berlin, Heidelberg, 1st edn. (2010)

\bibitem{PhysRevLett.126.140502}
Harrow, A.W., Napp, J.C.: Low-depth gradient measurements can improve
  convergence in variational hybrid quantum-classical algorithms. Phys. Rev.
  Lett.  \textbf{126},  140502 (Apr 2021).
  \doi{10.1103/PhysRevLett.126.140502},
  \url{https://link.aps.org/doi/10.1103/PhysRevLett.126.140502}

\bibitem{ref2}
Havlíček, V., Córcoles, A.D., Temme, K., Harrow, A.W., Kandala, A., Chow,
  J.M., Gambetta, J.M.: {Supervised learning with quantum-enhanced feature
  spaces}. Nature  \textbf{567}(7747),  209--212 (March 2019).
  \doi{10.1038/s41586-019-0980-2},
  \url{https://ideas.repec.org/a/nat/nature/v567y2019i7747d10.1038_s41586-019-0980-2.html}

\bibitem{HoeflerHT23}
Hoefler, T., H{\"{a}}ner, T., Troyer, M.: Disentangling hype from practicality:
  On realistically achieving quantum advantage. Commun. {ACM}  \textbf{66}(5),
  82--87 (2023). \doi{10.1145/3571725}, \url{https://doi.org/10.1145/3571725}

\bibitem{Kondor10}
Kondor, R.: A {Fourier} space algorithm for solving quadratic assignment
  problems. In: Charikar, M. (ed.) Proceedings of the Twenty-First Annual
  {ACM-SIAM} Symposium on Discrete Algorithms, {SODA} 2010, Austin, Texas, USA,
  January 17-19, 2010. pp. 1017--1028. {SIAM} (2010).
  \doi{10.1137/1.9781611973075.82},
  \url{https://doi.org/10.1137/1.9781611973075.82}

\bibitem{10.1007/3-540-12868-9_99}
Lazard, D.: Gr{\"o}bner bases, {Gaussian} elimination and resolution of systems
  of algebraic equations. In: van Hulzen, J.A. (ed.) Computer Algebra. pp.
  146--156. Springer Berlin Heidelberg, Berlin, Heidelberg (1983)

\bibitem{Montanaro2019}
Montanaro, A.: Quantum-walk speedup of backtracking algorithms. Theory of
  Computing  \textbf{14},  15:1--24 (2019). \doi{10.4086/toc.2018.v014a015}

\bibitem{NeumannWitt2022}
Neumann, F., Witt, C.: Runtime analysis of the (1+1) {EA} on weighted sums of
  transformed linear functions. In: Rudolph, G., Kononova, A.V., Aguirre, H.,
  Kerschke, P., Ochoa, G., Tu{\v{s}}ar, T. (eds.) Parallel Problem Solving from
  Nature -- PPSN XVII. pp. 542--554. Springer International Publishing, Cham
  (2022)

\bibitem{ref1}
Nielsen, M.A., Chuang, I.L.: Quantum Computation and Quantum Information.
  Cambridge University Press (2000)

\bibitem{ref3}
Peruzzo, A., McClean, J., Shadbolt, P., Yung, M.H., Zhou, X.Q., Love, P.J.,
  Aspuru-Guzik, A., O’Brien, J.L.: A variational eigenvalue solver on a
  photonic quantum processor. Nature communications  \textbf{5},  4213:1--7
  (2014)

\bibitem{doi:10.1137/S0097539795293172}
Shor, P.W.: Polynomial-time algorithms for prime factorization and discrete
  logarithms on a quantum computer. SIAM Journal on Computing  \textbf{26}(5),
  1484--1509 (1997). \doi{10.1137/S0097539795293172},
  \url{https://doi.org/10.1137/S0097539795293172}

\bibitem{TintosWC15}
Tin{\'{o}}s, R., Whitley, L.D., Chicano, F.: Partition crossover for
  pseudo-boolean optimization. In: He, J., Jansen, T., Ochoa, G., Zarges, C.
  (eds.) Proceedings of the 2015 {ACM} Conference on Foundations of Genetic
  Algorithms XIII, Aberystwyth, United Kingdom, January 17 - 20, 2015. pp.
  137--149. {ACM} (2015). \doi{10.1145/2725494.2725497},
  \url{https://doi.org/10.1145/2725494.2725497}

\bibitem{PhysRevX.10.021067}
Zhou, L., Wang, S.T., Choi, S., Pichler, H., Lukin, M.D.: Quantum approximate
  optimization algorithm: Performance, mechanism, and implementation on
  near-term devices. Phys. Rev. X  \textbf{10},  021067 (Jun 2020).
  \doi{10.1103/PhysRevX.10.021067},
  \url{https://link.aps.org/doi/10.1103/PhysRevX.10.021067}

\end{thebibliography}

\end{document}